\documentclass[11pt]{amsart}

\usepackage{mathptmx}
\usepackage{amsfonts}
\usepackage{amssymb}
\usepackage{amsmath}
\usepackage{latexsym}
\usepackage{amsthm}
\usepackage{stmaryrd}

\newtheorem{theorem}{Theorem}
\newtheorem{definition}{Definition}
\newtheorem{lemma}{Lemma}
\newtheorem{remark}{Remark}

\newtheorem{example}{Example}

\newcommand{\ra}{\rightarrow}
\newcommand{\R}{{\mathbb R}}
\newcommand{\C}{{\mathbb C}}
\newcommand{\Z}{{\mathbb Z}}
\newcommand{\N}{{\mathbb N}}
\newcommand{\K}{{\mathcal K}}
\newcommand{\Cc}{{\mathcal C}}
\newcommand{\ds}{\displaystyle}
\newcommand{\poisson}{\mathfrak{p}}
\newcommand{\Gammab}{\mathsf{\Gamma}}
\newcommand{\konsi}{\mathfrak{K}}

\usepackage{tikz}
\usetikzlibrary{arrows}
\tikzset{
type1/.style={circle,fill=white!100,inner sep=1.5pt},
type2/.style={circle,draw=white,fill=black!100,inner sep=1.5pt},
shorten >=.5pt
}

\title{The necessity of wheels in universal quantization formulas}
\author{Giuseppe Dito}
\address{Institut de Math{\'e}matiques de Bourgogne,
Universit{\'e} de Bourgogne,
B.P. 47870,
21078 Dijon Cedex,
France}
\email{giuseppe.dito@u-bourgogne.fr}
\urladdr{http://monge.u-bourgogne.fr/gdito}

\newcommand{\mone}{
\begin{aligned}
m_1(f,g)&=
\begin{tikzpicture}[baseline=12pt]
\node[type1] (f) at (0,0)  {};
\node[type1] (g) at (1,0)  {};
\node[type2] (1) at (.5,1)  {};
\node[below] at (f) {$f$};
\node[below] at (g) {$g$};
\draw[-latex',thick] (1) to (f);
\draw[-latex',densely dotted,thick] (1) to (g);
\end{tikzpicture}
\end{aligned}
}

\newcommand{\mtwo}{
\begin{aligned}
m_2(f,g)&=\frac{1}{2}\ 
\begin{tikzpicture}[baseline=12pt]
\node[type1] (f) at (0,0)  {};
\node[type1] (g) at (1,0)  {};
\node[type2] (1) at (0,1)  {};
\node[type2] (2) at (1,1)  {};
\node[below] at (f) {$f$};
\node[below] at (g) {$g$};
\draw[-latex',thick] (1) to (f);
\draw[-latex',densely dotted,thick] (1) to (g);
\draw[-latex',thick] (2) to (f);
\draw[-latex',thick,densely dotted] (2) to (g);
\end{tikzpicture}
+ \frac{1}{3}\ 
\begin{tikzpicture}[baseline=12pt]
\node[type1] (f) at (0,0)  {};
\node[type1] (g) at (1,0)  {};
\node[type2] (1) at (0,1)  {};
\node[type2] (2) at (1,1)  {};
\node[below] at (f) {$f$};
\node[below] at (g) {$g$};
\draw[-latex',thick] (1) to (f);
\draw[-latex',densely dotted,thick] (1) to (2);
\draw[-latex',thick] (2) to (f);
\draw[-latex',thick,densely dotted] (2) to (g);
\end{tikzpicture}
+ \frac{1}{3}\ 
\begin{tikzpicture}[baseline=12pt]
\node[type1] (f) at (0,0)  {};
\node[type1] (g) at (1,0)  {};
\node[type2] (1) at (0,1)  {};
\node[type2] (2) at (1,1)  {};
\node[below] at (f) {$f$};
\node[below] at (g) {$g$};
\draw[-latex',thick] (1) to (f);
\draw[-latex',densely dotted,thick] (1) to (g);
\draw[-latex',thick] (2) to (1);
\draw[-latex',thick,densely dotted] (2) to (g);
\end{tikzpicture}
\end{aligned}
}

\newcommand{\mthree}{
\begin{aligned}
m_3(f,g)&=
\frac{1}{6}\ 
\begin{tikzpicture}[baseline=12pt]
\node[type1] (f) at (0,0)  {};
\node[type1] (g) at (1.5,0)  {};
\node[type2] (1) at (0,1)  {};
\node[type2] (2) at (.75,1)  {};
\node[type2] (3) at (1.5,1) {};
\node[below] at (f) {$f$};
\node[below] at (g) {$g$};
\draw[-latex',thick] (1) to (f);
\draw[-latex',densely dotted,thick] (1) to(g);
\draw[-latex',thick] (2) to (f);
\draw[-latex',thick,densely dotted] (2) to (g);
\draw[-latex',thick] (3) to (f);
\draw[-latex',thick,densely dotted] (3) to (g);
\end{tikzpicture}
+\frac{1}{3}\ 
\begin{tikzpicture}[baseline=12pt]
\node[type1] (f) at (0,0)  {};
\node[type1] (g) at (1.5,0)  {};
\node[type2] (1) at (0,1)  {};
\node[type2] (2) at (.75,1)  {};
\node[type2] (3) at (1.5,1) {};
\node[below] at (f) {$f$};
\node[below] at (g) {$g$};
\draw[-latex',thick] (1) to (f);
\draw[-latex',densely dotted,thick] (1) to(2);
\draw[-latex',thick] (2) to (f);
\draw[-latex',thick,densely dotted] (2) to (g);
\draw[-latex',thick] (3) to (f);
\draw[-latex',thick,densely dotted] (3) to (g);
\end{tikzpicture}
+\frac{1}{3}\ 
\begin{tikzpicture}[baseline=12pt]
\node[type1] (f) at (0,0)  {};
\node[type1] (g) at (1.5,0)  {};
\node[type2] (1) at (0,1)  {};
\node[type2] (2) at (.75,1)  {};
\node[type2] (3) at (1.5,1) {};
\node[below] at (f) {$f$};
\node[below] at (g) {$g$};
\draw[-latex',thick] (1) to (f);
\draw[-latex',densely dotted,thick] (1) to(g);
\draw[-latex',thick] (2) to (f);
\draw[-latex',thick,densely dotted] (2) to (g);
\draw[-latex',thick] (3) to (2);
\draw[-latex',thick,densely dotted] (3) to (g);
\end{tikzpicture}\\
&+\frac{1}{6}\ 
\begin{tikzpicture}[baseline=12pt]
\node[type1] (f) at (0,0)  {};
\node[type1] (g) at (1.5,0)  {};
\node[type2] (1) at (0,1)  {};
\node[type2] (2) at (.75,1)  {};
\node[type2] (3) at (1.5,1) {};
\node[below] at (f) {$f$};
\node[below] at (g) {$g$};
\draw[-latex',thick] (1) to (f);
\draw[-latex',densely dotted,thick] (1) to[out=30,in=150] (3);
\draw[-latex',thick] (2) to (f);
\draw[-latex',thick,densely dotted] (2) to (g);
\draw[-latex',thick] (3) to (2);
\draw[-latex',thick,densely dotted] (3) to (g);
\end{tikzpicture}
+\frac{1}{6}\ 
\begin{tikzpicture}[baseline=12pt]
\node[type1] (f) at (0,0)  {};
\node[type1] (g) at (1.5,0)  {};
\node[type2] (1) at (0,1)  {};
\node[type2] (2) at (.75,1)  {};
\node[type2] (3) at (1.5,1) {};
\node[below] at (f) {$f$};
\node[below] at (g) {$g$};
\draw[-latex',thick] (1) to (f);
\draw[-latex',densely dotted,thick] (1) to(2);
\draw[-latex',thick] (2) to (f);
\draw[-latex',thick,densely dotted] (2) to (g);
\draw[-latex',thick] (3) to[out=150,in=30] (1);
\draw[-latex',thick,densely dotted] (3) to (g);
\end{tikzpicture}
+\frac{1}{3}\ 
\begin{tikzpicture}[baseline=12pt]
\node[type1] (f) at (0,0)  {};
\node[type1] (g) at (1.5,0)  {};
\node[type2] (1) at (0,1)  {};
\node[type2] (2) at (.75,1)  {};
\node[type2] (3) at (1.5,1) {};
\node[below] at (f) {$f$};
\node[below] at (g) {$g$};
\draw[-latex',thick] (1) to (f);
\draw[-latex',densely dotted,thick] (1) to[out=30,in=150] (3);
\draw[-latex',thick] (2) to[out=200,in=-20]  (1);
\draw[-latex',thick,densely dotted] (2) to[out=-20,in=200] (3);
\draw[-latex',thick] (3) to (f);
\draw[-latex',thick,densely dotted] (3) to (g);
\end{tikzpicture}\\
&+\frac{1}{3}\ 
\begin{tikzpicture}[baseline=12pt]
\node[type1] (f) at (0,0)  {};
\node[type1] (g) at (1.5,0)  {};
\node[type2] (1) at (0,1)  {};
\node[type2] (2) at (.75,1)  {};
\node[type2] (3) at (1.5,1) {};
\node[below] at (f) {$f$};
\node[below] at (g) {$g$};
\draw[-latex',thick] (1) to (f);
\draw[-latex',densely dotted,thick] (1) to (g);
\draw[-latex',thick] (2) to[out=200,in=-20] (1);
\draw[-latex',thick,densely dotted] (2) to[out=-20,in=200] (3);
\draw[-latex',thick] (3) to[out=150,in=30] (1);
\draw[-latex',thick,densely dotted] (3) to (g);
\end{tikzpicture}
+\frac{1}{6}\ 
\begin{tikzpicture}[baseline=12pt]
\node[type1] (f) at (0,0)  {};
\node[type1] (g) at (1.5,0)  {};
\node[type2] (1) at (0,1)  {};
\node[type2] (2) at (.75,1)  {};
\node[type2] (3) at (1.5,1) {};
\node[below] at (f) {$f$};
\node[below] at (g) {$g$};
\draw[-latex',thick] (1) to (f);
\draw[-latex',densely dotted,thick] (1) to[out=30,in=150] (3);
\draw[-latex',thick] (2) to (f);
\draw[-latex',thick,densely dotted] (2) to[out=-20,in=200] (3);
\draw[-latex',thick] (3) to (f);
\draw[-latex',thick,densely dotted] (3) to (g);
\end{tikzpicture}
+\frac{1}{6}\ 
\begin{tikzpicture}[baseline=12pt]
\node[type1] (f) at (0,0)  {};
\node[type1] (g) at (1.5,0)  {};
\node[type2] (1) at (0,1)  {};
\node[type2] (2) at (.75,1)  {};
\node[type2] (3) at (1.5,1) {};
\node[below] at (f) {$f$};
\node[below] at (g) {$g$};
\draw[-latex',thick] (1) to (f);
\draw[-latex',densely dotted,thick] (1) to (g);
\draw[-latex',thick] (2) to[out=200,in=-20] (1);
\draw[-latex',thick,densely dotted] (2) to (g);
\draw[-latex',thick] (3) to[out=150,in=30] (1);
\draw[-latex',thick,densely dotted] (3) to (g);
\end{tikzpicture}
\end{aligned}
}

\begin{document}

\begin{abstract}
In the context of formal deformation quantization, we provide an elementary argument showing that 
any universal quantization formula necessarily involves graphs with wheels.
\end{abstract}

\maketitle

\section{Introduction} \label{sec:intro}
Since its inception in the 70's by Flato and his 
collaborators~\cite{BFFLS1}, deformation quantization has undergone several major developments.
The most spectacular one being the proof by Kontsevich~\cite{K2003}
of the existence of star-products on any smooth Poisson manifold as a consequence
of his Formality theorem. This theorem establishes the existence of an $L_\infty$ quasi-isomorphism
between two differential graded Lie algebras naturally attached to a manifold, namely, the Hochschild complex endowed with Gerstenhaber bracket on one side, and the space of polyvector fields
endowed with the Schouten-Nijenhuis bracket and trivial differential, on the other side.

The proof relies on the discovery of an explicit formula (in the affine case) 
in terms of Feynman graphs and their weights, and corresponds to the perturbation 
series of a Poisson sigma model~\cite{CF}.
An important feature of Kontsevich's quantization formula is that it is a universal quantization formula
in the sense that it only depends on the Poisson structure and its derivatives of all orders.
Moreover, among the family of graphs considered by Kontsevich there are graphs with wheels, i.e. having 
oriented cycles as subgraphs. The presence of graphs with wheels
discards direct generalizations of Kontsevich formula to infinite-dimensional
spaces (quantum field theory) as their computation involves traces leading to
an ill-defined star-product. Of course for special classes of Poisson brackets it is still possible 
to find quantization formulas
which do not involve wheels. This is the case for the well-known Moyal product or
the BCH-quantization~\cite{berezin,drinfeld,gutt} on the dual of a Lie algebra, 
but these examples do not reflect the general situation.

In this paper, we show that a universal quantization formula, i.e. a correspondence that associates
to any Poisson bracket defined on the affine space a star-product, 
does necessarily involves graphs with wheels.  The proof consists essentially in reducing the general case
to the obstruction found by Penkava and Vanhaecke~\cite{PV2000} in their study of the quantization
of polynomial Poisson brackets.

The paper is organized as follows. Section~2 is devoted to preliminaries: basic notions on deformation quantization
and Kontsevich formula are briefly recalled there. In Section~3,  we study universal 2-cocycles 
in the Hoschchild complex and show that the second cohomology space of the subcomplex consisting of universal
cochains without wheels is one-dimensional and is generated by the Poisson bivector. The main theorem
is proved in Section~4. 

I am indebted to T.~Willwacher for letting me know another derivation of the 
main result of this paper. His proof~\cite{Will} 
does not rely on the obstruction of Penkava and Vanhaecke.

\section{Preliminaries} 
Let $X$ be a smooth $d$-dimensional manifold. 
Let $A=C^\infty(X)$ be the algebra of smooth complex-valued
functions on $X$ with product $m_0\colon A\times A\rightarrow A$, that is,
the usual product of smooth functions on $X$. 

\subsection{Deformation quantization}
Recall that a Poisson bivector $\poisson$ on $X$ is a section
in $\Gammab(\wedge^{2} TX)$ such that 
the Poisson bracket associated to $\poisson$
$$
\{f,g\}=\langle\poisson,df\wedge dg\rangle, \quad f,g\in A,
$$
endows $A$ with a structure of Lie algebra satisfying the Leibniz property.
In local coordinates $(x^1,\ldots,x^d)$ the Poisson bracket reads
$$
\{f,g\}= \sum_{1\leq i,j\leq d}\poisson^{ij}\partial_i f \partial_j g,
$$
where $\partial_k$ denotes the partial derivative with respect to $x^k$.

Let $A\llbracket\hbar\rrbracket$ be the space of formal series in $\hbar$ with coefficients in $A$.
A star-product~\cite{BFFLS1} is a $\C\llbracket\hbar\rrbracket$-bilinear product on $A\llbracket\hbar\rrbracket$
which is an associative deformation of the algebra of smooth functions on $X$:
\begin{equation}\label{eq:star}
\ds c_\hbar = m_0 + \sum_{k\geq1} \hbar^k c_k,
\end{equation}
where the $c_k$ are bidifferential operators vanishing on constants
and such that 
the antisymmetric part of $c_1$ is equal to the Poisson 
bivector~$\poisson$, i.e., a star-product is a noncommutative associative deformation of the
pointwise product of functions in the direction
of the Poisson bivector.

\subsection{Hochschild cohomology}
The normalized differential Hochschild cochain complex 
of the associative algebra $A$ with 
values in itself 
$\Cc^\bullet(A,A)=\oplus_{m\geq0} \Cc^m(A,A)$ 
consists of polydifferential operators on $X$ that are vanishing on constants.  
Locally, an $m$-cochain $C\in\Cc^m(A,A)$ has the form
\begin{equation}\label{cochain}
C(f_1,\ldots,f_m) = \sum C_{\alpha_1\cdots\alpha_m} \partial^{\alpha_1}f_1\cdots \partial^{\alpha_m}f_m,\quad\quad f_1,\ldots,f_m\in A,
\end{equation}
where the sum is finite and runs over multi-indices 
$\alpha_i\in\N^d$ such that $\vert \alpha_i\vert \geq 1$, and the $C_{\alpha_1\cdots\alpha_m}$
are locally defined smooth functions on $X$. 

As usual, we shall consider the Hochschild 
complex as a  $\Z$-graded vector space with a shift in the degree:
$D^\bullet_\mathrm{poly}(X)=\Cc^\bullet(A,A)[1]$. Hence
$$
D^k_\mathrm{poly}(X) = 
\begin{cases} 
\Cc^{k+1}(A,A) &\text{for } k\geq -1, \\
\{0\} &\text{otherwise}. 
\end{cases}
$$

The Gerstenhaber bracket $[\cdot,\cdot]_G$ on 
$D^\bullet_\mathrm{poly}(X)$ is defined on homogeneous elements
$D_i \in D^{k_i}_\mathrm{poly}(X)$ by:
$$
[D_1, D_2]_G = D_1\circ D_2 - (-1)^{k_1 k_2}D_2\circ D_1,
$$
where $\circ\colon D^{k_1}_\mathrm{poly}(X)\times D^{k_2}_\mathrm{poly}(X)
\rightarrow D^{k_1+k_2}_\mathrm{poly}(X)$ 
is a composition law for polydifferential operators:
\begin{equation*}
\begin{split}
(D_1\circ D_2)&(f_0,\ldots, f_{k_1+k_2}) \\
&=
\sum_{0\leq j\leq k_1} (-1)^{jk_2} D_1(f_0,\ldots, f_{j-1}, D_2(f_j,
\ldots,f_{j+k_2}), f_{j+k_2+1},\ldots f_{k_1+k_2}).
\end{split}
\end{equation*}

The Hochschild differential is  given by $\delta = [m_0,\ \cdot\ ]_G$.
It is the standard Hochschild differential $d$ up to a sign: 
for $D\in D^{k}_\mathrm{poly}(\R^d)$, 
we have $\delta D=(-1)^k d D$.
Recall that $(D^\bullet_\mathrm{poly}(X), [\cdot,\cdot]_G, \delta)$ 
is a differential graded Lie algebra.
\begin{remark}
In terms of the Gertenhaber bracket, the associativity of the 
product~(\ref{eq:star}) is equivalent to the Maurer-Cartan equations:
\begin{equation}\label{eq:mc}
\delta c_k + \frac{1}{2} \sum_{\substack{a+b=k \\ a,b\geq1}} [c_a,c_b]_G =0, \quad\text{ for all } 
k\geq1.
\end{equation}
A deformation $c_\hbar$ is said to define an associative deformation up to order $r$,
if (\ref{eq:mc}) is satisfied for $k=1,2,\ldots,r$.
\end{remark}

The cohomology $H^\bullet(A,A)$ of the complex $(\Cc^\bullet(A,A),\delta)$ 
is the space of  polyvectors
$\Gammab(\wedge^\bullet TX)$. This is a smooth version of the Hochschild-Kostant-Rosenberg (HKR)
theorem due to Vey~\cite{Vey}, still called HKR theorem. We denote by 
$T^\bullet_\mathrm{poly}(X)$ the graded vector space 
$T^\bullet_\mathrm{poly}(X)=\oplus_{k\in\mathbb{Z}}T^k_\mathrm{poly}(X)$, 
where 
$$
T^k_\mathrm{poly}(X) = 
\begin{cases} 
\Gammab(\wedge^{k+1} TX) &\text{for } k\geq -1, \\
\{0\} &\text{otherwise}. 
\end{cases}
$$
$T^\bullet_\mathrm{poly}(X)$
is endowed with the Schouten-Nijenhuis bracket $[\cdot,\cdot]_{SN}$,
On decomposable tensors it is given by:
\begin{equation}\label{schouten}
\begin{split}
[&\xi_0\wedge\cdots\wedge\xi_k,\eta_0\wedge\cdots\wedge\eta_l]_{SN}\\
&=\sum_{0\leq i\leq k}\ \sum_{0\leq j\leq l} (-1)^{i+j} [\xi_i,\eta_j]
\wedge\xi_0 \wedge\cdots\wedge\hat\xi_i\wedge\cdots\wedge\xi_k\wedge
\eta_0\wedge\cdots\wedge\hat\eta_j\wedge\cdots\wedge\eta_l.
\end{split}
\end{equation}
Recall that the Jacobi identity for a bivector $\pi \in \Gammab(\wedge^{2} TX)$
is equivalent to the condition $[\pi,\pi]_{SN}=0$.

$(T^\bullet_\mathrm{poly}(X),[\cdot,\cdot]_{SN})$ 
is a graded Lie algebra that is considered as a differential
graded Lie algebra with trivial differential~$\mathbf{0}$.

The Formality theorem~\cite{K2003} establishes the existence of 
an $L_\infty$ quasi-isomorphism between the differential graded Lie algebras 
$(T^\bullet_\mathrm{poly}(X),[\cdot,\cdot]_{SN},\mathbf{0})$ 
and $(D^\bullet_\mathrm{poly}(X), [\cdot,\cdot]_G, \delta)$.
A remarkable consequence of the Formality theorem is the existence
of deformation quantization of any smooth Poisson manifold: 
\begin{theorem}[\cite{K2003}]
Let $\poisson$ be a Poisson bivector on a smooth manifold $X$. Then
there exists a star-product $c_\hbar$ on $X$ such that 
$\frac{1}{2}\big(c_\hbar(f,g) - c_\hbar(g,f)\big) 
= \hbar \{f,g\} + O(\hbar^2)$.
\end{theorem}
Actually, Kontsevich gives an explicit description of the 
$L_\infty$ quasi-isomorphism for $X=\R^d$ in terms of graphs and weights.
We shall briefly recall a few notions on graphs that we will need in our discussion
and refer to \cite{K2003} for details.

\subsection{Graphs}

A simple directed graph is a graph whose edges are oriented 
and such that it does not contain loops:
\begin{tikzpicture}[baseline=0pt,scale=1]
\node[type2] (1) at (0,0)  {}; 
\draw[-latex',thick]  (1) to[out=140,in=20,loop] (1);
\end{tikzpicture}
or multiple edges:
\begin{tikzpicture}[baseline=0pt,scale=1]
\node[type2] (1) at (0,0)  {}; 
\node[type2] (2) at (.5,.4)  {};
\draw[-latex',thick]  (1) to[out=60,in=200] (2);
\draw[-latex',thick]  (1) to[out=0,in=250] (2);
\end{tikzpicture}
.
The indegree (resp.~outdegree) of a vertex is the number of edges ending (resp.~starting) 
at that vertex.

Following~\cite{K2003}, we consider a family of graphs (with a slight departure
from the convention of \cite{K2003}).

\begin{definition}\label{def:k}
The set $\K$ consists of all the simple directed graphs $\Gamma$ satisfying:
\begin{enumerate}
\item The set of vertices $V_\Gamma$ is finite and is a disjoint union of nonempty sets: 
$V_\Gamma=V_\Gamma^1 \sqcup V_\Gamma^2$.
Vertices belonging to $V_\Gamma^1$ (resp. $V_\Gamma^2$) are said
of type~1 (resp.~2);
\item Each vertex of type 1 is of outdegree 2;
\item Each vertex of type 2 is of indegree at least~1 and of outdegree 0;
\item Vertices and edges are labeled.
\end{enumerate}
We denote by $\K_{n,m}$ for $m,n\geq1$ the subset of $\K$ consisting of graphs having
$n$ vertices of type~1 and $m$ vertices of type~2. Thus a graph in $\K_{n,m}$ 
has $2n$ edges.
\end{definition}
We shall make a slight abuse of notation 
by dropping the labels on the vertices of type~1.

We set $X=\R^d$ ($d\geq 2$) and
$(x^1,\ldots,x^d)$ designates a coordinate system on $X$.
Let $\poisson$ be a smooth Poisson bivector on $X$.  The Poisson 
bracket of two smooth functions  $f,g$ is graphically represented by
a graph in $\K_{1,2}$:
\begin{gather*}
\{f,g\}=\sum_{1\leq i,j \leq d} \poisson^{ij}\;\partial_i\; f \partial_j g\ =
\begin{tikzpicture}[baseline=12pt,scale=.8]
\node[type1] (f) at (0,0)  {};
\node[type1] (g) at (1,0)  {};
\node[type2] (1) at (.5,1)  {};
\node[below] at (f) {$f$};
\node[below] at (g) {$g$};
\draw[-latex',thick] (1) to (f);
\draw[-latex',densely dotted,thick] (1) to (g);
\end{tikzpicture}.
\end{gather*}

More generally, one associates a cochain $B^\poisson_\Gamma$ in $\Cc^m(A,A)$
to any graph $\Gamma\in \K_{n,m}$ by letting the edges of $\Gamma$
act by partial derivatives (see \cite{K2003} for details.)
It is convenient to denote vertices of type~1 by $\bullet$ and vertices of type~2 
by letters $f,g,\ldots$ that will eventually be representing functions in $A$.
The labels of the edges will be graphically specified  by solid and dotted lines.
As an example, to the graph $\Gamma\in \K_{4,3}$ below one associates
$B^\poisson_\Gamma\in \Cc^3(A,A)$ by
\begin{equation}\label{tridiff}
B^\poisson_\Gamma(f,g,h)=
\begin{tikzpicture}[baseline=0pt,scale=1]
\node[type1] (f) at (0,0)  {};
\node[type1] (g) at (1,0)  {};
\node[type1] (h) at (2,0)  {};
\node[type2] (1) at (0.25,1)  {};
\node[type2] (2) at (1,1)  {};
\node[type2] (3) at (1.75,1) {};
\node[type2] (4) at (.7,.6) {};
\node[below] at (f) {$f$};
\node[below] at (g) {$g$};
\node[below] at (h) {$h$};
\draw[-latex',thick] (1) to (f);
\draw[-latex',thick] (3) to (h);
\draw[-latex',densely dotted,thick] (2) to (g);
\draw[-latex',thick,densely dotted] (1) to[out=-20,in=200] (2);
\draw[-latex',thick] (2) to[out=-20,in=200] (3);
\draw[-latex',thick] (4) to (1);
\draw[-latex',thick,densely dotted] (4) to (g);
\draw[-latex',thick,densely dotted] (3) to[out=160,in=20] (2);
\end{tikzpicture}
=\sum\;
\poisson^{i_1j_1}\;
\partial_{i_1}\poisson^{i_2j_2}\;
\partial_{j_2j_3}\poisson^{i_3j_3}\;
\partial_{i_3}\poisson^{i_4j_4}\;
\partial_{i_2}f\;
\partial_{j_1j_4}g\;
\partial_{i_4}h,
\end{equation}
where the sum runs over ${1\leq i_1,j_1, i_2,j_2, i_3,j_3, i_4,j_4\leq d}$.
Notice that we have implicitly labeled the vertices of type~1 in the sum, but
any other labeling leads to the same cochain.

Kontsevich formula associates
to a Poisson bivector $\poisson$ on $\R^d$ an associative deformation of $(A,m_0)$ 
\begin{equation}\label{konsi}
\konsi_\hbar = m_0 + \sum_{n\geq1}\hbar^n k_n^\poisson(f,g),
\quad \text{where}\quad k_n^\poisson(f,g)=\sum_{\Gamma\in \K_{n,2}} w(\Gamma) \ B^\poisson_\Gamma,
\end{equation}
and where $w(\Gamma)\in\R$ is the Kontsevich's weight of the graph $\Gamma$,
it is independent of the Poisson bivector and the dimension $d$.

\medskip
\noindent{\bf Convention.} Sometimes it is immaterial to distinguish between
the first and second argument of the Poisson bracket. In such situations we shall
simply represent the Poisson bracket by using solid lines for both edges, i.e., 
\begin{tikzpicture}[baseline=5pt,scale=.6]
\node[type1] (f) at (0,0)  {};
\node[type1] (g) at (1,0)  {};
\node[type2] (1) at (.5,1)  {};
\draw[-latex',thick] (1) to (f);
\draw[-latex',thick] (1) to (g);
\end{tikzpicture}.
Moreover, when it is not essential to  fully draw a graph, but just keep  few vertices
and edges relevant to our discussion, we shall use a ``grey zone'' to indicate that
not all the vertices or edges in that zone are drawn.  
This is illustrated by
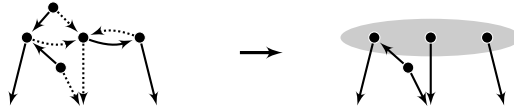
\begin{center}
\begin{tikzpicture}[baseline=12pt,scale=1]
\node[type1] (f) at (0,0)  {};
\node[type1] (g) at (1,0)  {};
\node[type1] (h) at (2,0)  {};
\node[type2] (1) at (0.25,1)  {};
\node[type2] (2) at (1,1)  {};
\node[type2] (3) at (1.75,1) {};
\node[type2] (4) at (.7,.6) {};
\node[type2] (5) at (.6,1.4) {};
\draw[-latex',thick] (5) to (1);
\draw[-latex',thick,densely dotted] (5) to (2);
\draw[-latex',thick] (1) to (f);
\draw[-latex',thick] (3) to (h);
\draw[-latex',densely dotted,thick] (2) to (g);
\draw[-latex',thick,densely dotted] (1) to[out=-20,in=200] (2);
\draw[-latex',thick] (2) to[out=-20,in=200] (3);
\draw[-latex',thick] (4) to (1);
\draw[-latex',thick,densely dotted] (4) to (g);
\draw[-latex',thick,densely dotted] (3) to[out=160,in=20] (2);
\node[type1] (k1) at (3,.8)  {};
\node[type1] (k2) at (3.75,.8)  {};
\draw[-latex',very thick] (k1) to (k2);
\end{tikzpicture}
\hskip5mm
\begin{tikzpicture}[baseline=12pt,scale=1]
\fill[fill=black!20] (1,1) ellipse [x radius=1.2cm,y radius=.25cm];
\node[type1] (f) at (0,0)  {};
\node[type1] (g) at (1,0)  {};
\node[type1] (h) at (2,0)  {};
\node[type2] (1) at (0.25,1)  {};
\node[type2] (2) at (1,1)  {};
\node[type2] (3) at (1.75,1) {};
\node[type2] (4) at (.7,.6) {};
\draw[-latex',thick] (1) to (f);
\draw[-latex',thick] (3) to (h);
\draw[-latex',thick] (2) to (g);
\draw[-latex',thick] (4) to (1);
\draw[-latex',thick] (4) to (g);
\end{tikzpicture}
\end{center}

\section{Universal 2-cocycles} \label{sec:universal}
Recall that a directed cycle in a simple directed graph
consists of a subgraph  for which
each vertex is of indegree and outdegree~1.
We say that a simple directed graph is with wheels
if it has a directed cycle as subgraph. An elementary result from
graph theory says that a simple directed graph such that each vertex
is of outdegree at least~1 must have a wheel. In the rest of this paper, we set $X=\R^d$
and $A=C^\infty(\R^d)$.
 
\begin{definition} A cochain $C\in \Cc^m(A,A)$, $m\geq1$, is called $\poisson$-universal
if it is a finite $\R$-linear combination of polydifferential operators associated to graphs
in $\K_{n,m}$, i.e., $C$ is a finite sum of the form
\begin{equation}\label{eq:diff}
C = \sum_{n\geq1}\ \sum_{\Gamma\in \K_{n,m}} a^{}_\Gamma \ B^\poisson_\Gamma,
\quad a^{}_\Gamma\in\R.
\end{equation}
A cochain $C\in \Cc^m(A,A)$, $m\geq1$, 
is called $\poisson$-universal without wheels
if all the graphs $\Gamma$ appearing in (\ref{eq:diff}) have no wheels.
\end{definition}

\begin{example}
The 2-cochain $C_1(f,g)=$
\begin{tikzpicture}[baseline=0pt,scale=.8]
\node[type1] (f) at (0,0)  {};
\node[type1] (g) at (1.5,0)  {};
\node[type2] (1) at (0,1)  {};
\node[type2] (2) at (.75,1)  {};
\node[type2] (3) at (1.5,1) {};
\node[below] at (f) {$f$};
\node[below] at (g) {$g$};
\draw[-latex',thick] (1) to (f);
\draw[-latex',densely dotted,thick] (1) to (g);
\draw[-latex',thick] (2) to[out=200,in=-20] (1);
\draw[-latex',thick,densely dotted] (2) to[out=-20,in=200] (3);
\draw[-latex',thick] (3) to[out=150,in=30] (1);
\draw[-latex',thick,densely dotted] (3) to (g);
\end{tikzpicture}
is $\poisson$-universal without wheels, while
$C_2(f,g)=$
\begin{tikzpicture}[baseline=0pt,scale=.8]
\node[type1] (f) at (0,0)  {};
\node[type1] (g) at (1.5,0)  {};
\node[type2] (1) at (0,1)  {};
\node[type2] (2) at (.75,1)  {};
\node[type2] (3) at (1.5,1) {};
\node[below] at (f) {$f$};
\node[below] at (g) {$g$};
\draw[-latex',thick] (1) to (f);
\draw[-latex',densely dotted,thick] (2) to (g);
\draw[-latex',thick,densely dotted] (1) to[out=-20,in=200] (2);
\draw[-latex',thick] (2) to[out=-20,in=200] (3);
\draw[-latex',thick] (3) to[out=150,in=30] (1);
\draw[-latex',thick,densely dotted] (3) to (g);
\end{tikzpicture}
is $\poisson$-universal with wheels.
\end{example}

\begin{definition}
A universal quantization formula is a correspondence  that associates
a star-product to a Poisson bivector
$$
\poisson\mapsto c_\hbar^\poisson = m_0 + \sum_{k\geq1} \hbar^k c_k^\poisson,
$$
where the $c_k^\poisson$ are $\poisson$-universal 2-cochains.
\end{definition}
To keep the notation simple, we will omit the reference to $\poisson$ in
$c_k^\poisson$ and write $c_k$, etc.
\begin{example}
Kontsevich formula~(\ref{konsi}) for a star-product is an example of 
universal quantization formula \textit{with wheels}. Wheels are already appearing
at order 2 in $\hbar$, the 2nd term being explicitly given by~\cite{Dito}
$$
k_2(f,g)=\frac{1}{2}\ 
\begin{tikzpicture}[baseline=12pt]
\node[type1] (f) at (0,0)  {};
\node[type1] (g) at (1,0)  {};
\node[type2] (1) at (0,1)  {};
\node[type2] (2) at (1,1)  {};
\node[below] at (f) {$f$};
\node[below] at (g) {$g$};
\draw[-latex',thick] (1) to (f);
\draw[-latex',densely dotted,thick] (1) to (g);
\draw[-latex',thick] (2) to (f);
\draw[-latex',thick,densely dotted] (2) to (g);
\end{tikzpicture}
+ \frac{1}{3}\ 
\begin{tikzpicture}[baseline=12pt]
\node[type1] (f) at (0,0)  {};
\node[type1] (g) at (1,0)  {};
\node[type2] (1) at (0,1)  {};
\node[type2] (2) at (1,1)  {};
\node[below] at (f) {$f$};
\node[below] at (g) {$g$};
\draw[-latex',thick] (1) to (f);
\draw[-latex',densely dotted,thick] (1) to (2);
\draw[-latex',thick] (2) to (f);
\draw[-latex',thick,densely dotted] (2) to (g);
\end{tikzpicture}
+ \frac{1}{3}\ 
\begin{tikzpicture}[baseline=12pt]
\node[type1] (f) at (0,0)  {};
\node[type1] (g) at (1,0)  {};
\node[type2] (1) at (0,1)  {};
\node[type2] (2) at (1,1)  {};
\node[below] at (f) {$f$};
\node[below] at (g) {$g$};
\draw[-latex',thick] (1) to (f);
\draw[-latex',densely dotted,thick] (1) to (g);
\draw[-latex',thick] (2) to (1);
\draw[-latex',thick,densely dotted] (2) to (g);
\end{tikzpicture}
- \frac{1}{6}\ 
\begin{tikzpicture}[baseline=12pt]
\node[type1] (f) at (0,0)  {};
\node[type1] (g) at (1,0)  {};
\node[type2] (1) at (0,1)  {};
\node[type2] (2) at (1,1)  {};
\node[below] at (f) {$f$};
\node[below] at (g) {$g$};
\draw[-latex',thick] (1) to (f);
\draw[-latex',densely dotted,thick] (1) to[out=-20,in=200] (2);
\draw[-latex',thick] (2) to[out=160,in=20] (1);
\draw[-latex',thick,densely dotted] (2) to (g);
\end{tikzpicture}
$$
\end{example}

\begin{example}\label{ex:pv}
In~\cite{PV2000}, Penkava and Vanhaecke have constructed a universal
quantization formula without wheels associative {up to order 3}, denoted here by
$m_\hbar = m_0 + \hbar m_1 + \hbar^2 m_2 + \hbar^3 m_3$, where
\begin{align*}
&\mone \\
& \\
&\mtwo \\
& \\
&\mthree
\end{align*}
For a general Poisson bivector, it was shown in~\cite{PV2000} that $m_\hbar$ cannot 
be extended to a 4th order associative deformation. 
\end{example}

\begin{lemma}\label{lem:wheels}
Let $\Gamma$ be a graph in $\K_{n,1}$ with $n\geq2$, then $\Gamma$ has at least a wheel.
\end{lemma}
\begin{proof}
Let us call $f$ the vertex of type~2 in $\Gamma$. Since $\Gamma$ is a simple graph, 
edges ending at $f$ must originate from different vertices of type~1.
By removing the vertex $f$ and all the edges
ending at $f$ we get a simple directed graph $\Gamma'$. This truncation is pictorially represented as:
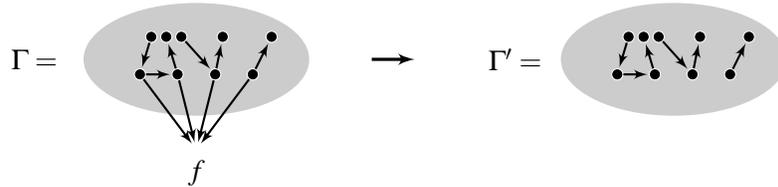
\begin{center}
\begin{tikzpicture}[baseline=12pt]
\fill[black!20] (.75,1.2) ellipse [x radius=1.5cm,y radius=.75cm];
\node[type1] (gam1) at (-1.85,1.2)  {};
\node[right] at (gam1) {$\Gamma=$};
\node[type1] (f) at (.75,0)  {};
\node[type2] (1) at (0,1)  {};
\node[type2] (2) at (.5,1)  {};
\node[type2] (3) at (1,1)  {};
\node[type2] (4) at (1.5,1)  {};
\node[type2] (f1) at (.15,1.5)  {};
\node[type2] (f2) at (.35,1.5)  {};
\node[type2] (f22) at (.55,1.5)  {};
\node[type2] (f3) at (1.1,1.5)  {};
\node[type2] (f4) at (1.75,1.5)  {};
\node[below] at (f) {$f$};
\draw[-latex',thick] (1) to (f);
\draw[-latex',thick] (1) to (2);
\draw[-latex',thick] (2) to (f);
\draw[-latex',thick] (3) to (f);
\draw[-latex',thick] (4) to (f);
\draw[-latex',thick] (f1) to (1);
\draw[-latex',thick] (2) to (f2);
\draw[-latex',thick] (4) to (f4);
\draw[-latex',thick] (3) to (f3);
\draw[-latex',thick] (f22) to (3);
\node[type1] (k1) at (3,1.2)  {};
\node[type1] (k2) at (3.75,1.2)  {};
\draw[-latex',very thick] (k1) to (k2);
\end{tikzpicture}
\hskip5mm
\begin{tikzpicture}[baseline=12pt]
\node[type1] (gam1) at (-1.85,1.2)  {};
\node[right] at (gam1) {$\Gamma'=$};
\fill[fill=black!20] (.75,1.2) ellipse [x radius=1.5cm,y radius=.75cm];
\node[type1] (f) at (.75,0)  {};
\node[type2] (1) at (0,1)  {};
\node[type2] (2) at (.5,1)  {};
\node[type2] (3) at (1,1)  {};
\node[type2] (4) at (1.5,1)  {};
\node[type2] (f1) at (.15,1.5)  {};
\node[type2] (f2) at (.35,1.5)  {};
\node[type2] (f22) at (.55,1.5)  {};
\node[type2] (f3) at (1.1,1.5)  {};
\node[type2] (f4) at (1.75,1.5)  {};
\draw[-latex',thick] (1) to (2);
\draw[-latex',thick] (f1) to (1);
\draw[-latex',thick] (2) to (f2);
\draw[-latex',thick] (4) to (f4);
\draw[-latex',thick] (3) to (f3);
\draw[-latex',thick] (f22) to (3);
\end{tikzpicture}
\end{center}
Each vertex in $\Gamma'$ is
of outdegree at least~1, hence it contains a wheel.
\end{proof}
We have excluded the case $n=1$ in this Lemma since  $\K_{1,1}$ is empty.

\begin{lemma}\label{lem:tensor}
Let $\Lambda\in \Gammab(\wedge^{2} TX)\subset \Cc^2(A,A)$ be a bivector. 
If $\Lambda$ is $\poisson$-universal without wheels, 
then $\Lambda=a\, \poisson$ for some $a\in\R$.
\end{lemma}
\begin{proof}
Suppose $\Lambda$ is $\poisson$-universal, then it can be represented as a finite sum
$\sum_{n\geq1}\ \sum_{\Gamma\in \K_{n,2}} a^{}_\Gamma \ B^\poisson_\Gamma$.
Each graph $\Gamma$ appearing in this sum has two vertices of type~2
each of them having indegree~1.
We distinguish several cases.

1) If $\Gamma\in \K_{1,2}$, then clearly $B^\poisson_\Gamma = \pm \poisson$.

2) If $\Gamma\in \K_{n,2}$ with $n\geq 2$, then there are three possible subcases
to which we apply a similar truncation as the one used in Lemma~\ref{lem:wheels}:

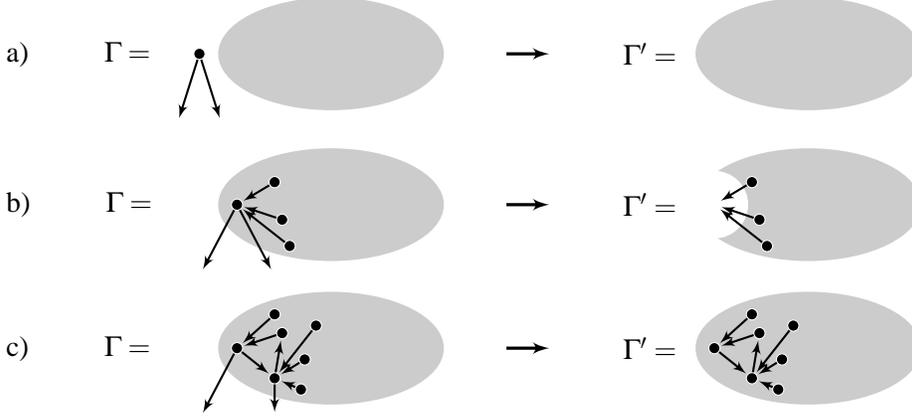
\medskip
\begin{tikzpicture}[baseline=12pt]
\fill[black!20] (.75,1.2) ellipse [x radius=1.5cm,y radius=.75cm];
\node[type1] (gam1) at (-3.7,1.2)  {};
\node[right] at (gam1) {a)\hskip10mm$\Gamma=$};
\node[type1] (f) at (-1.3,0.25)  {};
\node[type1] (g) at (-.7,0.25)  {};
\node[type2] (1) at (-1,1.2)  {};
\draw[-latex',thick] (1) to (f);
\draw[-latex',thick] (1) to (g);
\node[type1] (k1) at (3,1.2)  {};
\node[type1] (k2) at (3.75,1.2)  {};
\draw[-latex',very thick] (k1) to (k2);
\end{tikzpicture}
\hskip5mm
\begin{tikzpicture}[baseline=12pt]
\node[type1] (gam1) at (-1.85,1.2)  {};
\node[right] at (gam1) {$\Gamma'=$};
\fill[black!20] (.75,1.2) ellipse [x radius=1.5cm,y radius=.75cm];
\end{tikzpicture}

\medskip
\begin{tikzpicture}[baseline=12pt]
\fill[black!20] (.75,1.2) ellipse [x radius=1.5cm,y radius=.75cm];
\node[type1] (gam1) at (-3.7,1.2)  {};
\node[right] at (gam1) {b)\hskip10mm$\Gamma=$};
\node[type1] (f) at (-1,0.25)  {};
\node[type1] (g) at (0,0.25)  {};
\node[type2] (1) at (-.5,1.2)  {};
\draw[-latex',thick] (1) to (f);
\draw[-latex',thick] (1) to (g);
\node[type2] (2) at (0,1.5)  {};
\node[type2] (3) at (.1,1)  {};
\node[type2] (4) at (.2,.65)  {};
\draw[-latex',thick] (2) to (1);
\draw[-latex',thick] (3) to (1);
\draw[-latex',thick] (4) to (1);
\node[type1] (k1) at (3,1.2)  {};
\node[type1] (k2) at (3.75,1.2)  {};
\draw[-latex',very thick] (k1) to (k2);
\end{tikzpicture}
\hskip5mm
\begin{tikzpicture}[baseline=12pt]
\node[type1] (gam1) at (-1.85,1.2)  {};
\node[right] at (gam1) {$\Gamma'=$};
\fill[fill=black!20] (.75,1.2) ellipse [x radius=1.5cm,y radius=.75cm];
\fill[white] (-.5,1.2) circle (.45);
\node[type1] (1) at (-.5,1.2)  {};
\node[type2] (2) at (0,1.5)  {};
\node[type2] (3) at (.1,1)  {};
\node[type2] (4) at (.2,.65)  {};
\draw[-latex',thick] (2) to (1);
\draw[-latex',thick] (3) to (1);
\draw[-latex',thick] (4) to (1);
\end{tikzpicture}

\smallskip
\begin{tikzpicture}[baseline=12pt]
\fill[black!20] (.75,1.2) ellipse [x radius=1.5cm,y radius=.75cm];
\node[type1] (gam1) at (-3.7,1.2)  {};
\node[right] at (gam1) {c)\hskip10mm$\Gamma=$};
\node[type1] (f) at (-1,0.25)  {};
\node[type1] (g) at (0,0.25)  {};
\node[type2] (1) at (-.5,1.2)  {};
\node[type2] (11) at (0,.8)  {};
\draw[-latex',thick] (1) to (f);
\draw[-latex',thick] (11) to (g);
\node[type2] (2) at (0,1.65)  {};
\node[type2] (3) at (.1,1.4)  {};
\node[type2] (4) at (.35,.65)  {};
\node[type2] (5) at (.4,1.05)  {};
\node[type2] (6) at (.55,1.5)  {};
\draw[-latex',thick] (2) to (1);
\draw[-latex',thick] (3) to (1);
\draw[-latex',thick] (4) to (11);
\draw[-latex',thick] (5) to (11);
\draw[-latex',thick] (6) to (11);
\draw[-latex',thick] (11) to (3);
\draw[-latex',thick] (1) to (11);
\node[type1] (k1) at (3,1.2)  {};
\node[type1] (k2) at (3.75,1.2)  {};
\draw[-latex',very thick] (k1) to (k2);
\end{tikzpicture}
\hskip5mm
\begin{tikzpicture}[baseline=12pt]
\node[type1] (gam1) at (-1.85,1.2)  {};
\node[right] at (gam1) {$\Gamma'=$};
\fill[fill=black!20] (.75,1.2) ellipse [x radius=1.5cm,y radius=.75cm];
\node[type2] (1) at (-.5,1.2)  {};
\node[type2] (11) at (0,.8)  {};
\node[type2] (2) at (0,1.65)  {};
\node[type2] (3) at (.1,1.4)  {};
\node[type2] (4) at (.35,.65)  {};
\node[type2] (5) at (.4,1.05)  {};
\node[type2] (6) at (.55,1.5)  {};
\draw[-latex',thick] (2) to (1);
\draw[-latex',thick] (3) to (1);
\draw[-latex',thick] (4) to (11);
\draw[-latex',thick] (5) to (11);
\draw[-latex',thick] (6) to (11);
\draw[-latex',thick] (11) to (3);
\draw[-latex',thick] (1) to (11);
\node[type1] (k1) at (3,1.2)  {};
\node[type1] (k2) at (3.75,1.2)  {};
\end{tikzpicture}

\smallskip
For each subcase, when it occurs, the graph $\Gamma'$ contains at least a wheel. Indeed
for a), each vertex of~$\Gamma'$ has outdegree~2; for b), it follows 
Lemma~\ref{lem:wheels}; for c) each vertex of $\Gamma'$ has 
outdegree at least~1.

In conclusion, a bivector which is $\poisson$-universal without wheels can only
be proportional to the Poisson bivector.
\end{proof}

\begin{lemma}\label{lem:coboundary}
Let $T\in \Cc^1(A,A)$ be such
that its coboundary $\delta T$ is 
$\poisson$-universal without wheels. Then $\delta T=0$.
\end{lemma}
\begin{proof}
By getting rid of an irrelevant derivation in $T$, one can consider that
$T$ is a $\poisson$-universal 1-cochain. Indeed,
one can use a homotopy formula (see e.g.~\cite{OMY} for 
an explicit formula for 2-cocycles) to define a $\poisson$-universal 1-cochain
$T'$ such that $\delta T' = \delta T$. Since $\delta T$ is without wheels so will
be $T'$, but Lemma~\ref{lem:wheels} tells us that if $T'$ is not vanishing, 
then it must contain a wheel. Therefore $T'=0$ and $\delta T=0$.
\end{proof}

By the HKR theorem, any cocycle $C\in \Cc^m(A,A)$ is a sum $C=\Lambda +\delta C'$,
where $\Lambda\in \Gammab(\wedge^{m} TX)$ and $C'\in \Cc^{m-1}(A,A)$. 
The $m$-vector $\Lambda$ being the antisymmetric part of $C$ then, if $C$ is $\poisson$-universal without wheels, so are $\Lambda$ and $\delta C'$.
As an immediate 
consequence of Lemmas~\ref{lem:tensor} and~\ref{lem:coboundary} we have:
\begin{lemma}\label{cor:cocycle}
Let $C\in \Cc^2(A,A)$ be a Hochschild 2-cocycle which
is $\poisson$-universal without wheels, then $C=a\,\poisson$
for some $a\in\R$.
\end{lemma}

Notice that the Hochschild coboundary of a  $\poisson$-universal cochain is
$\poisson$-universal and similarly for $\poisson$-universal cochains without wheels.
Hence one can consider the subcomplex of the Hochschild complex consisting of
$\poisson$-universal cochains without wheels.  In essence, Lemma~\ref{cor:cocycle}
states that the second cohomology space for this subcomplex is one-dimensional
and is generated by $\poisson$.

\section{The necessity of wheels} \label{sec:result}

Here we establish the main result of this paper by showing
that the terms up to order 3 in $\hbar$ of 
any universal quantization formula without wheels are, up to a change of the deformation parameter,
given by the $m_i$'s in Example~\ref{ex:pv} and hence get a contradiction.

\begin{theorem}\label{th:nogo}
Any universal quantization formula involves graphs
with wheels.
\end{theorem}
\begin{proof}
We proceed by contradiction assuming that $\ds s_\hbar = m_0 + \sum_{k\geq1} \hbar^k s_k$ 
is a universal quantization formula without wheels.

Associativity implies $\delta s_1=0$ and Lemma~\ref{cor:cocycle} gives that
$s_1 = a_1\, \poisson$ for some $a_1\in\R$. The normalization of the star-product,
$\frac{1}{2}\big(s_1(f,g)-s_1(g,f)\big)=\{f,g\}$, forces $a_1=1$, hence $s_1=m_1=\poisson$.

At second order in $\hbar$, the associativity of $s_\hbar$ implies 
$$
\delta s_2 + \frac{1}{2}[s_1,s_1]_G=0, \ \text{ i.e., }\ \delta s_2 + \frac{1}{2}[m_1,m_1]_G=0.
$$

Since $\frac{1}{2}[m_1,m_1]_G=-\delta m_2$,
we find that $s_2-m_2$ is a 2-cocycle and it is $\poisson$-universal without wheels. By 
Lemma~\ref{cor:cocycle} we have $s_2= m_2+ a_2 m_1$ for some real number $a_2$.

At third order in $\hbar$, we have
$$
\delta s_3 + [s_2,s_1]_G=0,\ \text{ i.e., }\ \delta s_3 + [m_2 + a_2 m_1,m_1]_G=0.
$$
{}From $[m_2,m_1]_G=-\delta m_3$ and $[m_1,m_1]_G= -2\delta m_2$
we deduce that $s_3 - m_3 - 2 a_2 m_2$ is a $\poisson$-universal 2-cocycle
without wheels. Using  Lemma~\ref{cor:cocycle} again, we find that
$s_3=m_3 + 2 a_2 m_2 + a_3 m_1$ for some real number $a_3$.

Therefore the first terms in $s_\hbar$ are necessarily of the form:
\begin{align*}
s_1&=m_1 ,\\
s_2&=m_2 + a_2 m_1, \\
s_3&=m_3 + 2 a_2 m_2 + a_3 m_1.
\end{align*}

The change of parameter $\hbar \ra \hbar - a_2 \hbar^2 - (a_3 -2 a_2^2) \hbar^3$ in $s_\hbar$ 
gives us a new universal quantization formula $s_\hbar'$ such that $s_i' = m_i$ for $i=1,2,3$.
But $m_0 + \hbar m_1 + \hbar^2 m_2 + \hbar^3 m_3$  is an associative deformation up to order 3 
that cannot be extended to order 4 (cf. Example~\ref{ex:pv}). 
Hence we have reached a contradiction.
\end{proof}

\end{document}